\documentclass[12pt,letterpaper,reqno]{amsart}
\usepackage[includeheadfoot,left=1in,right=1in,top=1in, bottom=1in]{geometry}
\usepackage[utf8]{inputenc}
\usepackage{enumerate}
\usepackage{amsmath,amsthm,amssymb,amsfonts,latexsym}
\usepackage{dsfont}
\usepackage[gen]{eurosym}
\usepackage[german,english]{babel}
\usepackage{mathtools}
\usepackage{comment,soul}
\usepackage{mathrsfs}
\usepackage{latexsym}
\usepackage{tikz}					% Grafiken zeichnen

\newcommand{\Tr}[1]{\mathrm{tr}^{(#1)}}

\DeclareMathOperator{\supp}{supp}

\DeclarePairedDelimiterX\braket[2]{\langle}{\rangle}{#1 | #2}
\DeclarePairedDelimiterX\bra[1]{\langle}{|}{#1}
\DeclarePairedDelimiterX{\braketop}[3]{\langle}{\rangle}{#1 \delimsize\vert #2 \delimsize\vert #3}
\DeclarePairedDelimiterX\ketbra[1]{\lvert}{\rvert}{#1 \rangle \langle #1}
\DeclarePairedDelimiterX\ket[1]{|}{\rangle}{#1}

\newcommand{\projectionpsi}{P_{\ket{\psi}}}
\newcommand{\projectionPsiN}{P_{\ket{\Psi_N^{1,\omega}}}}
\newcommand{\projectionvarphi}{P_{\ket{\varphi}}} %{\ketbra{\varphi}}
\newcommand{\projectionvarphiN}{P_{\ket{\varphi_N}}}%
\newcommand{\projectionvarphiNjomega}[1]{P_{\ket{\varphi_N^{#1,\omega}}}}%{\ketbra{\varphi_N^{j,\omega}}
\newcommand{\projectionvarphiNtwice}{P_{\ket{\varphi_N,\varphi_N}}}%{\projectionvarphiNjomegatwice
\newcommand{\projectionpsipsi}{P_{\ket{\psi,\psi}}}

\newcommand{\jparticlehilbertspaceAbrev}[1]{\mathcal H_{N}^{#1,\omega}}
\newcommand{\jparticlehilbertspaceSymAbrev}[1]{\mathcal H_{N,\text{s}}^{#1,\omega}}
\newcommand{\oneparticlehilbertspaceAbrev}{\mathcal H_N^{1,\omega}}

\newtheorem{theorem}{Theorem}%[section]
\newtheorem{lemma}[theorem]{Lemma}

\newtheorem{remark}[theorem]{Remark}
\newtheorem{corollary}[theorem]{Corollary}

\newtheorem{example}[theorem]{Example}
\newcommand{\be}{\begin{equation}}
\newcommand{\ee}{\end{equation}}
\newcommand{\bea}{\begin{eqnarray*}}
	\newcommand{\eea}{\end{eqnarray*}}
\newcommand{\beq}{\begin{eqnarray}}
\newcommand{\eeq}{\end{eqnarray}}

\newcommand{\tr}{\mathrm{tr}^{(1)}}
\newcommand{\Trtwo}{\mathrm{tr}^{(2)}}
\newcommand{\densitymatrix}[1]{\varrho^{(#1)}_{\Psi_N^{1,\omega}}}
\newcommand{\densitymatrixgen}[1]{\varrho^{(#1)}_{\Psi_N}}

\newcommand{\N}{\varrho^{(N)}_{\Psi_N^{1,\omega}}}

\allowdisplaybreaks

\usepackage{todonotes}

\title[Existence and absence of BEC in the interacting KL-model]{
Existence and absence of Bose--Einstein condensation in the  interacting random Kac--Luttinger model
} 

\subjclass[2010]{}


\author[C.~Boccato]{Chiara Boccato}
\address{Department of Mathematics, Università di Pisa, 56127 Pisa, Italy}
\email{chiara.boccato@unipi.it}

\author[J.~Kerner]{Joachim Kerner}
\address{FernUniversit\"at in Hagen, Fakultät für Mathematik und Informatik, Lehrgebiet Analysis, 58084 Hagen, Germany}
\email{joachim.kerner@fernuni-hagen.de}

\author[M.~Pechmann]{Maximilian Pechmann}
\address{Department of Mathematics, Tennessee Technological University, Cookeville, TN 38505, USA}
\email{mpechmann@tntech.edu}

\author[W.~Spitzer]{Wolfgang Spitzer}
\address{FernUniversit\"at in Hagen, Fakultät für Mathematik und Informatik, 58084 Hagen, Germany}
\email{wolfgang.spitzer@fernuni-hagen.de}

\date{\today}

\thanks{
}

\begin{document}

\begin{abstract}
In this paper, we study  interacting bosons at zero temperature in a random and higher-dimensional continuum model introduced by Kac and Luttinger. For weak interactions we prove that there is condensation in the lowest eigenstate of the one-particle Hamiltonian (type-I BEC). For strong interactions, however, we show that condensation in a localized state cannot occur. We also prove generalized condensation, where a family of eigenstates of the one-particle Hamiltonian is macroscopically occupied as a whole. Combining these results yields a scenario where there is generalized condensation into a family of eigenstates of the one-particle Hamiltonian, but none of them is macroscopically occupied itself (type-III BEC). This proves a transition in the type of condensation. To the best of our knowledge, this is the first rigorous result in this direction for a random continuum model in higher dimensions. 
\end{abstract}

\maketitle

\section{Introduction}
\noindent 
This paper deals with Bose--Einstein condensation (BEC) in random environments and since BEC behaves very differently in random settings, we aim to contribute to a better understanding thereof. Generally speaking, we study interacting many-body quantum systems of bosonic particles, i.e., particles that are described by permutation-symmetric wavefunctions. In quantum mechanics, BEC constitutes an important example of a phase transition that is typically hard to prove mathematically. This is true, in particular, for a random setting. 

In the following, we study interacting bosonic particles in a random model originally introduced by Kac and Luttinger in the seventies~\cite{kac1973bose,kac1974bose}. The impurities in the Kac--Luttinger model (KL-model) are hard balls of a fixed radius that are randomly distributed in space according to a Poisson point process. Kac and Luttinger concluded that there exists a finite critical particle density and they also conjectured that only the one-particle ground state is macroscopically occupied, similar to the non-interacting (ideal) Bose gas in three dimensions as discussed by Einstein~\cite{EinsteinBECI,EinsteinBECII}. This conjecture became known as the Kac--Luttinger conjecture and was resolved only very recently in \cite{SznitmanKAC} based on previous results established in \cite{KERNER2020287}. Compared to the non-random setting, the difficulty lies in controlling the eigenvalues, and even more so the spectral gaps between eigenvalues due to the intricate geometry of the (random) one-particle configuration space. At this point, it is worth referring to the one-dimensional analog of the KL-model, the so-called Luttinger--Sy model~\cite{LuttSyEnergy}. Here, one randomly distributes points on the real axis and places Dirac delta potentials at these points. Due to the one-dimensional nature of this model, eigenvalue considerations become more feasible but remain complex. Regarding the non-interacting Bose gas in the Luttinger--Sy model, one could highlight that the analog of the Kac--Luttinger conjecture for this model was resolved in \cite{LenZag}, see also \cite{jaeck2010nature}; an alternative proof was given in~\cite{KERNER2020287}. 

Of course, taking into account inter-particle interactions is desirable but complicates the discussion of BEC and this is true also in the non-random setting. There, in suitable (scaling) regimes, advances in the last decades led to a variety of rigorous results; see, for example, \cite{LiebSeiringer2002, ZagLauVer,SeiRev,BBCS2019,BBCSoptimalrate,SchleinRev, BaiVogel, ZagrebnovReviewBEC} and references therein. However, since randomness leads to localization effects (Anderson localization), the transition from the non-interacting to the interacting Bose gas is much more drastic in the random setting. On the one hand, randomness enhances BEC in the non-interacting Bose gas due to the presence of Lifshitz tails (fast decay of the integrated density of states at the bottom of the spectrum); on the other hand, the condensate is expected to be spatially localized (for example, on a logarithmic scale) and this spatial localization then poses problems as soon as inter-particle interactions are implemented. Generally speaking, it is this competition between peculiar spectral effects such as Lifshitz tails (being related to relatively large spectral gaps) and (Anderson) localization that makes the study of BEC in random environments special. 

Due to what has been said above, proving absence of BEC becomes, in random settings, equally important to proving existence. As a matter of fact, a main aim of this paper is to establish some results on the absence of condensation into states that are too localized (Theorem~\ref{AbsenceofLocalizedBEC}). Similar results for the KL-model in the non-percolation regime (at positive temperature) have been established in \cite{kerner_pechmann_2023} and in \cite{kerner2021effect} for the Luttinger--Sy model. Regarding BEC in interacting Bose gases in the KL-model, one should also refer to \cite{boccato2024interacting}: in this paper, the authors use results from \cite{SznitmanKAC} in order to prove macroscopic occupation of a one-particle state (the minimizer of a Hartree-type functional) for certain two-particle interactions that are weak enough (see also Theorem~\ref{TheoremMacOcc}). This shows that the (localized) condensate of the non-interacting Bose gas is nevertheless stable against some weak enough two-particle interactions. However, comparing the results on existence and absence of BEC, one question immediately comes up: Should one expect a macroscopically occupied one-particle state in the KL-model for strong interactions? Or, to phrase it differently, what type of condensation can exist in an interacting system with strong interactions and strong localization? In this context, let us refer to \cite{SeiringerWarzel} where absence of BEC was proved for the Tonks-Girardeau gas and to \cite{kerner2021effect} where absence of BEC was established for the Luttinger--Sy model.

The goal of this paper is to approach these questions. After introducing the KL-model in Section~\ref{Section The setting} and other notions in Section~\ref{Section Preliminaries}, we establish Theorem~\ref{TheoremMacOcc} in Section~\ref{FirstSection} which shows that the one-particle ground state is also macroscopically occupied in the interacting system given the two-particle interaction is weak enough. However, since the one-particle ground state is presumably localized (it is localized in the non-percolation regime on a logarithmic scale), we derive Theorem~\ref{AbsenceofLocalizedBEC} in Section~\ref{SectionAbsence} which shows that a sufficiently localized state cannot be macroscopically occupied for strong enough two-particle interactions. In other words, the macroscopic occupation described by Theorem~\ref{TheoremMacOcc} will be lost for sufficiently strong two-particle interactions. In Section~\ref{SectionGeneralized} we then study condensation also for stronger two-particle interactions and establish Theorem~\ref{GenBECWeak} which proves generalized condensation. Hence, combining the results from Section~\ref{SectionAbsence} and Section~\ref{SectionGeneralized} then yields an example of an interacting Bose gas in the KL-model for which there is generalized condensation into a family of one-particle eigenstates but each individual state is not macroscopically occupied due to its spatial localization; in other words, there is type-III Bose--Einstein condensation (see Corollary~\ref{CorGenBEC}). Taking into account that one has type-I BEC in the non-interacting Bose gas (actually, only the one-particle ground state is macroscopically occupied), our paper establishes a transition in the nature of the condensate, by increasing the interaction strength, going from type I to type III; here, one should remark that similar results have been obtained in \cite{kerner2019bose} for the one-dimensional Luttinger--Sy model. However, to the best of our knowledge, this is the first time for which such a transition is proved for a random higher-dimensional model. We finish the paper with a discussion of our results in Section~\ref{SectionDiscussion}.
\section{The interacting random Kac--Luttinger model}\label{Section The setting}

\noindent To introduce our model, we start with a probability space $(\Omega, \mathcal{F}, \mathds{P})$ and a Poisson point process $\xi$ on $\mathds{R}^d$, $d \ge 2$, with constant intensity $\nu>0$, generating the points $\{x_m^\omega : m\in\mathds{N}\}$. In more detail, $\xi$ is a (measurable) map from this probability space into a set of point measures on $\mathds{R}^d$ such that for any Borel subset $\Lambda\subset \mathds{R}^d$ of finite Lebesgue volume $|\Lambda|$, the random variable $\Omega\ni\omega\mapsto \xi(\omega)(\Lambda)$ is Poisson distributed with mean $\nu |\Lambda|$. This means that $\mathds{P}\big(\omega\in\Omega : \xi(\omega)(\Lambda) = k\big)  = \frac{(\nu |\Lambda|)^k}{k!} \exp (-\nu |\Lambda|)$ for $k\in\mathds{N}_0$. Moreover, for any finite number $n$ of disjoint Borel subsets $\Lambda_1,\ldots,\Lambda_n$, the random variables $\omega\mapsto \xi(\omega)(\Lambda_1),\ldots,\omega\mapsto \xi(\omega)(\Lambda_n)$ are independent. We identify the random measure $\xi(\omega)$ with its support, namely the points $\{x_m^\omega : m\in\mathds{N}\}$.

Let $B_r(x)$ denote the closed ball of fixed radius $r>0$ centered at $x \in \mathds{R}^d$. 
Furthermore, for each $N \in \mathds{N}$, we define the box

\begin{equation*}
\Lambda_N \coloneqq \left(-L_N/2,+L_N/2 \right)^d \subset \mathds{R}^d \quad \text{where} \quad L_N := \rho^{-1/d} N^{1/d}
\end{equation*}
and $\rho >0$. Thus, for each $N \in \mathds N$,
\begin{equation*}
\Lambda_N^{\omega} \coloneqq \Lambda_N \backslash \bigcup_{m \in \mathds N} B_r(x_m^{\omega})\ , \quad \omega \in \Omega\ , 
\end{equation*}
is a random (open) set in $\mathds R^d$, which we refer to as the \emph{vacancy set}. 
\begin{remark}
    The vacancy consists $\mathds P$-almost surely and for all $N \in \mathds N$ of only finitely  many components \cite[Proposition 4.1]{meester1996continuum}. In this paper, we shall always refer to such a \textit{typical} $\omega \in \Omega$. Along the way, we will also subsume more properties that hold almost surely when referring to a (typical) $\omega \in \Omega$. 
\end{remark}
For each $N \in \mathds{N}$, we denote the number of components of the vacancy set $\Lambda_N^{\omega}$ by $K_N^{\omega} \in \mathds{N}_0$ (with the understanding that $K_N^{\omega} = 0$ if $\Lambda_N^{\omega} = \emptyset$). Defining the index set $\mathcal{K}_N^{\omega} := \{1, \ldots, K_N^{\omega} \}$ if $K_N^{\omega} \ge 1$ (and $\mathcal{K}_N^{\omega} := \emptyset$ if $K_N^{\omega} = 0$), we label the components by $k \in \mathcal{K}_N^{\omega}$. Therefore, if $K_N^{\omega} \ge 1$, each component of $\Lambda_N^{\omega}$ is denoted by $\Lambda_N^{k,\omega}$ where $k \in \mathcal{K}_N^{\omega}$, and the collection $\{\Lambda_N^{k,\omega}\}_{k \in \mathcal{K}_N^{\omega}}$ forms a partition of $\Lambda_N^{\omega}$. We remark that $\mathds P$-almost surely and for all but finitely many $N$, $\mathcal K_N^{\omega} \neq \emptyset$, see for example \cite[Chapter 4, Theorem 4.6]{sznitman1998brownian}. Therefore, in the following, we will always assume that $\omega \in \Omega$ and $N \in \mathds{N}$ are such that $\mathcal K_N^{\omega} \neq \emptyset$.

It is also important to recall that one distinguishes two regimes in the KL-model: the percolation and the non-percolation regime~\cite{sznitman1998brownian}. In the percolation regime, the intensity $\nu$ is smaller than some critical value and there exists an unbounded component of $\mathds{R}^d\backslash \bigcup_{m \in \mathds{N}} B_r(x_m^{\omega})$; in the non-percolation regime, the intensity is larger than this critical value and all components of $\mathds{R}^d\backslash \bigcup_{m \in \mathds{N}} B_r(x_m^{\omega})$ are bounded.

To simplify the notation, we write $\jparticlehilbertspaceSymAbrev{1} \coloneq \jparticlehilbertspaceAbrev{1} \coloneq L^2(\Lambda_N^{\omega})$ and, for $j \in \{2,\ldots,N\}$,
\begin{equation*}
\jparticlehilbertspaceAbrev{j} \coloneq L^2((\Lambda^{\omega}_N)^j) \quad \text{ and } \quad \jparticlehilbertspaceSymAbrev{j} \coloneq L^2_{\text{s}}((\Lambda^{\omega}_N)^j) \ .
\end{equation*}
Here, the index $\text{s}$ refers to the symmetric subspace of $L^2((\Lambda_N^{\omega})^{j})$, $j \in \{2,\ldots,N\}$.

Our system is then defined by the random, self-adjoint $N$-particle Hamiltonian
\begin{equation} \label{N particle Hamiltonian}
H_{N}^{\omega} \coloneqq -\sum_{j=1}^{N} \Delta_{j} + W_N^{(N)} \ ,
\end{equation}
which acts on the $N$-particle Hilbert space $\jparticlehilbertspaceSymAbrev{N}$, subject to Dirichlet boundary conditions. This Schrödinger operator is defined through its associated quadratic form with form domain $\mathcal{D}[H_N^{\omega}] = H_0^1((\Lambda_N^{\omega})^{N}) \cap \jparticlehilbertspaceSymAbrev{N}$, where $H_0^1(\cdot)$ is the usual Sobolev space.

In particular, the self-adjoint operator $-\sum_{j=1}^{N} \Delta_{j}$ is a standard lift of $-\Delta$ to $\jparticlehilbertspaceAbrev{N}$, which, after an appropriate restriction also defines an operator on $\jparticlehilbertspaceSymAbrev{N}$, and $-\Delta$ is the (self-adjoint) Dirichlet Laplacian on $\oneparticlehilbertspaceAbrev$. Furthermore, we introduce
\begin{equation*}
 W_N^{(N)} : \jparticlehilbertspaceSymAbrev{N} \to \jparticlehilbertspaceSymAbrev{N} \ , \ \Psi(x_1,\dots,x_N) \mapsto \sum_{1 \leq i < j \leq N}w_N(x_i-x_j) \Psi(x_1,\dots,x_N) \ ,
\end{equation*}
where $w_N \in (L^1 \cap L^{\infty})(\mathds{R}^d)$ is a non-negative, bounded, and even function.  We also introduce the corresponding operator
\begin{equation*}
    W_N^{(2)} : \jparticlehilbertspaceSymAbrev{2} \to \jparticlehilbertspaceSymAbrev{2} \ , \  \Psi(x_1,x_2) \mapsto w_N(x_1-x_2) \Psi(x_1,x_2) \ ,
\end{equation*}
on $\jparticlehilbertspaceSymAbrev{2}$.

From a physical point of view, the Hamiltonian \eqref{N particle Hamiltonian} 
describes a system of $N$ bosons confined to the box $\Lambda_N$. Each particle is described by the one-particle operator  
\begin{equation*}
    -\Delta+V^\omega\ ,
\end{equation*}
where the random external potential can be informally written as
\begin{equation*}
 V^\omega(x) \coloneqq 
\sum_{m \in \mathds{N}} \infty \cdot 1_{B_r(x_m^\omega)}(x) \ ,
\end{equation*}
and $-\Delta$ is the Dirichlet Laplacian on $L^2(\Lambda_N)$. The repulsive two-body interaction between the bosons is described by the 
non-negative, bounded function $w_N:\mathds{R}^d \rightarrow \mathds{R}_+$. 
Also, the Hamiltonian $H_N^{\omega}$ represents the total energy of the system: 
the operator $-\sum_{j=1}^{N} \Delta_j$ corresponds to the kinetic energy, 
while the operator $W_N^{(N)}$ accounts for the potential energy arising from the two-body interparticle interactions. 
The constant $\rho>0$ appearing in the definition of $\Lambda_N$ represents 
the particle density; consequently, taking the limit $N \to \infty$ corresponds 
to the thermodynamic limit.

\begin{remark}
In this paper, we abbreviate $L^p$-norms by writing, for instance, $\|\cdot\|_{2}$ instead of $\|\cdot\|_{L^{2}(\Lambda_N^{\omega})}$. We also omit specifications when writing inner products. In addition, we adopt the following notation for the asymptotic behavior of sequences: Whenever $(a_N)_{N \in \mathds{N}}$ is a non-negative and $(b_N)_{N \in \mathds{N}}$ a positive sequence, we write $a_N \ll b_N$ if and only if $\lim_{N \to \infty} a_N/b_N = 0$, and $a_N \sim b_N$ if and only if there exist constants $c, C > 0$ such that $c a_N \leq b_N \leq C a_N$ for all but finitely many $N \in \mathds{N}$. We use the notation $a_N \lesssim b_N$ if and only if $a_N \ll b_N$ or $a_N \sim b_N$.
\end{remark}

\section{Preliminaries} \label{Section Preliminaries}
\noindent We start by recalling well-known properties of linear operators on Hilbert spaces; for more details we refer the reader, for example, to \cite{weidmann2012linear}. 

Let $\mathscr H$ be a complex, separable Hilbert space. Let $A$ be a linear operator on $\mathscr H$ with (dense) domain $\mathcal D(A)$. Then $\mathcal N(A) \coloneq \{f \in \mathcal D(A) : Af = 0\}$ is the nullspace of $A$ and $\mathcal R(A) \coloneq \{Af : f \in \mathcal D(A)\}$ is the image of~$A$. Furthermore, a bounded operator $P$ on $\mathscr H$ is called a projection on $\mathscr H$ if $P^2 = P$. A projection $P$ on $\mathscr H$ is called an orthogonal projection on $\mathscr H$ onto $\mathcal R(P)$ if $\mathcal R(P) \perp \mathcal N(P)$. Orthogonal projections are uniquely defined by their image, and for any closed subspace $\mathcal G$ of $\mathscr H$ there exists an orthogonal projection $P_{\mathcal G}$ with $\mathcal R(P_{\mathcal G}) = \mathcal G$. Orthogonal projections on Hilbert spaces are self-adjoint and positive. The identity operator $\mathds 1_{\mathscr H}$ is the orthogonal projection on $\mathscr H$ onto $\mathscr H$.

Let $\mathcal F,\mathcal G$ be closed subspaces of $\mathscr H$: If and only if $\mathcal F \subset \mathcal G$, the operator $P_{\mathcal G} - P_{\mathcal F}$ is an orthogonal projection on $\mathscr H$. Similarly, if and only if $\mathcal F \perp \mathcal G$, $P_{\mathcal F} + P_{\mathcal G}$ is a orthogonal projection on $\mathscr H$.

Let $A,B$ be two linear operators on $\mathscr H$. If $A,B$ are two Hilbert--Schmidt operators\footnote{A Hilbert-Schmidt operator is a compact operator with square-summable singular values.}, or $A$ is a bounded operator and $B$ is a trace class operator\footnote{A trace-class operator is a compact operator with summable singular values.}, then $AB$ is a trace-class operator and we have
\begin{equation*}\label{cyclic property trace}
    \mathrm{tr}_{\mathscr H}(AB) = \mathrm{tr}_{\mathscr H}(BA) \ .
\end{equation*}

Let $\varrho$ be a positive (and thus self-adjoint) trace-class operator on $\mathscr H$.
Then $\varrho$ has a unique square root $\varrho^{1/2}$ such that $\varrho = \varrho^{1/2} \varrho^{1/2}$. Moreover, $\varrho^{1/2}$ is then a Hilbert--Schmidt operator, and a positive (and thus self-adjoint) operator, on $\mathscr H$. Next, let $A$ be a bounded operator on $\mathscr H$. Then $A \varrho^{1/2}$ is also a Hilbert--Schmidt operator.  

Therefore, $A\varrho^{1/2} \varrho^{1/2}$ is a trace-class operator and
\begin{equation*}\label{trace cyclic property}
    \mathrm{tr}_{\mathscr H}(A \varrho) = \mathrm{tr}_{\mathscr H}(A \varrho^{1/2} \varrho^{1/2}) = \mathrm{tr}_{\mathscr H}(\varrho^{1/2} A \varrho^{1/2}) \ .
\end{equation*}

If $A$ is bounded and positive, then
\begin{equation}\label{trace positive operator density matrix is positive}
    \mathrm{tr}_{\mathscr H}(A \rho) = \mathrm{tr}_{\mathscr H}(\rho^{1/2} A \rho^{1/2}) = \sum\limits_{j \in \mathds N} \braketop{\rho^{1/2} \psi_j}{A}{\rho^{1/2} \psi_j}\ge 0 \ ,
\end{equation}
where $(\psi_j)_{j \in \mathds N}$ is any orthonormal basis of $\mathscr H$.
In particular, for any closed subspaces $\mathcal F, \mathcal G$ of $\mathscr H$ with $\mathcal F \subset \mathcal G$,
\begin{equation} \label{inequatliy tr F G subspaces}
    \mathrm{tr}_{\mathscr H} (P_{\mathcal F} \varrho) \le \mathrm{tr}_{\mathscr H} (P_{\mathcal G} \varrho) \ ,
\end{equation}
since $P_{\mathcal G} - P_{\mathcal F}$ is an orthogonal projection and hence a positive, bounded operator.

\begin{remark}
For a normalized $\psi \in \mathscr H$ we write $\projectionpsi$ for the orthogonal projection on $\mathscr H$ onto $\mathrm{span}(\psi)=\mathds{C}\cdot\psi$, $\projectionpsipsi$ for the orthogonal projection on $\mathscr H \otimes \mathscr H$ onto $\mathrm{span}(\psi \otimes \psi)=\mathds{C}\cdot(\psi \otimes \psi)$, etc. Also, throughout this work, we use $\braketop{\psi}{A}{\psi}$ rather than $(\psi,A\psi)$ for a self-adjoint operator $A$.
\end{remark}

The lowest eigenvalue $E_{\text{QM},N}^{1,\omega}$ of the Hamiltonian $H_N^{\omega}$, hence the $N$-particle ground state energy, is given by 
\begin{equation*}
E_{\text{QM},N}^{1,\omega} \coloneqq \inf \big\{ \braketop{\Psi}{ H_N^{\omega}}{\Psi} : \Psi \in \mathcal{D}[H_N^{\omega}], \|\Psi \|_{2} = 1 \big\} \ ,
\end{equation*}
where we understand $\braketop{\Psi}{H_N^{\omega}}{\Psi}$ in the form sense.
There exists a (symmetric) minimizer $\Psi_N^{1,\omega}$, which we also refer to as a ground state of $H_N^{\omega}$, satisfying
\begin{equation*}
 \braketop{\Psi_N^{1,\omega}}{H_N^{\omega}}{\Psi_N^{1,\omega}} = E_{\text{QM},N}^{1,\omega} \ .
\end{equation*}

The $N$-particle density matrix
\begin{equation*}
 \N \coloneq \projectionPsiN
\end{equation*}
is a positive trace-class operator on $\mathcal{H}^{N,\omega}_{N}$ with trace one,
\begin{equation*}
 \Tr{N}_{\jparticlehilbertspaceAbrev{N}}(\densitymatrix{N}) = \sum\limits_{j \in \mathds N} \braketop{\Psi_j}{\densitymatrix{N}}{\Psi_j} = \big\|\Psi_N^{1,\omega}\big\|_2^2 =1\ ,
\end{equation*}
where $(\Psi_j)_{j \in \mathds N}$ is an arbitrary orthonormal basis of $\mathcal{H}^{N,\omega}_{N}$ (of course, a similar construction works for any $\Psi_N \in \mathcal{H}^{N,\omega}_{N}$). The corresponding one-particle density matrix $\densitymatrix{1}$ is then defined as a positive, trace-class operator on \( \oneparticlehilbertspaceAbrev \) such that
\begin{equation*}
\braketop{\varphi}{\densitymatrix{1}}{\psi} = \sum\limits_{j \in \mathds N} \braketop{\varphi, \Phi_j}{\densitymatrix{N}}{\psi, \Phi_j}
\end{equation*}
for all $\varphi, \psi \in \oneparticlehilbertspaceAbrev$, where $(\Phi_j)_{j \in \mathds N}$ is an orthonormal basis of $\mathcal{H}^{N-1,\omega}_{N}$ and where we used the notation $\braketop{\varphi, \Phi_j}{\densitymatrix{N}}{\psi, \Phi_j} = (\varphi \otimes \Phi_j, \densitymatrix{N} (\psi \otimes \Phi_j))$. 

This one-particle density matrix is normalized, meaning its trace satisfies
\begin{equation*}
 \tr_{\jparticlehilbertspaceAbrev{1}}(\densitymatrix{1}) = \sum\limits_{j \in \mathds N} \braketop{\varphi_j}{\densitymatrix{1}}{\varphi_j}= 1 \ ,
\end{equation*}
where $(\varphi_j)_{j\in\mathds N}$ is an arbitrary orthonormal basis of $\oneparticlehilbertspaceAbrev$. The number of particles occupying a one-particle state $\varphi \in \oneparticlehilbertspaceAbrev$, given the $N$-particle system is in the $N$-particle ground state $\Psi_N^{1,\omega}$, is then given by
\begin{equation*}
n_{\Psi_N^{1,\omega}}^{\varphi} := N \, \tr_{\jparticlehilbertspaceAbrev{1}}\left(\projectionvarphi \, \densitymatrix{1}  \right) \ .
\end{equation*}
The prefactor $N$ in front of the trace ensures that the total number of particles occupying states $(\varphi_j)_{j\in\mathds{N}}$ of an orthonormal basis of $ \oneparticlehilbertspaceAbrev$  is equal to the total number of particles $N$.

Similarly, the corresponding two-particle density matrix $\densitymatrix{2}$ is defined as a positive, trace-class operator on \( \mathcal{H}^{2,\omega}_{N} \) such that
\begin{equation*}
\braketop{\Phi}{\densitymatrix{2}}{\Psi} = \sum\limits_{j \in \mathds N} \braketop{\Phi, \Xi_j}{\densitymatrix{N}}{\Psi, \Xi_j}
\end{equation*}
for all $\Phi, \Psi \in  \mathcal{H}^{2,\omega}_{N}$, where $(\Xi_j)_{j \in \mathds N}$ is a orthonormal basis for $\mathcal{H}^{N-2,\omega}_{N}$. Furthermore,
\begin{equation*}
 \Trtwo_{\jparticlehilbertspaceAbrev{2}}(\densitymatrix{2}) = \sum\limits_{j \in \mathds N} \braketop{\Phi_j}{\densitymatrix{2}}{\Phi_j}= 1 \ ,
\end{equation*}
where $(\Phi_j)_{j\in\mathds N}$ is an arbitrary orthonormal basis of $\mathcal{H}^{2,\omega}_{N}$. In a recursive manner, one also defines the $j$-particle density matrix $\densitymatrix{j}$ on $\mathcal{H}^{j,\omega}_{N}$ in this way. 

Lastly, we denote the (canonical) eigenfunctions of the (self-adjoint) Dirichlet Laplacian~$-\Delta$ on $\oneparticlehilbertspaceAbrev$ by $\varphi_N^{j,\omega}$, $j \in \mathds N$, and the corresponding eigenvalues by $e_N^{j,\omega}$, $j \in \mathds N$, arranged in increasing order and repeated according to their multiplicities ($\varphi_N^{1,\omega}$ will be referred to as the one-particle ground state). Note that the term \textit{canonical} refers to eigenfunctions that are supported on one component of the vacancy set only; in other words, for each component of the vacancy set, one picks an orthonormal basis of eigenfunctions and extends them by zero to all of the vacancy set.

In the following, we use the notation, $j \in \mathds N$,
\begin{equation} \label{occupation number}
n_N^{j,\omega} \coloneq n_{\Psi_N^{1,\omega}}^{\varphi_N^{j,\omega}} = N \, \tr_{\oneparticlehilbertspaceAbrev}\left(\projectionvarphiNjomega{j} \, \densitymatrix{1} \right) \ ,
\end{equation}
and call $n_N^{j,\omega}$ the occupation number of the eigenstate $\varphi_N^{j,\omega}$. As mentioned above, $\sum_{j=1}^\infty n_N^{j,\omega} =N$. Loosely speaking, a state $\varphi^{j,\omega}_N$ is called macroscopically occupied if the occupation numbers $n_N^{j,\omega}$ are of order $N$ as the system size $L_N$ tends to infinity as $N \rightarrow \infty$ (the precise formulations are stated in the theorems below).

\section{A first result: Condensation into the one-particle ground state}\label{FirstSection} 
\noindent The by now proved Kac--Luttinger conjecture asserts that in the non-interacting Bose gas in the KL-model, and in some suitable probabilistic sense, only the one-particle ground state is macroscopically occupied (at positive temperature). The aim of this section is to prove that this macroscopic occupation persists in the interacting Bose gas (at zero temperature), assuming the interaction strength is not too strong. 
\begin{theorem}[BEC in the one-particle ground state]\label{TheoremMacOcc} Assume the system is in the state $\N$, where $\Psi_N^{1,\omega}$ is a ground state of $H_{N}^{\omega}$ defined in  \eqref{N particle Hamiltonian}. Let $n_N^{1,\omega}$ be the occupation number of $\varphi_N^{1,\omega}$.  Then, the following holds:
\begin{enumerate}[$(\mathrm{i})$]
\item If
\begin{equation}\label{Assumption interaction strength type I BEC KLM}
\|w_N\|_{\infty} \ll \dfrac{1}{N(\ln N)^{1+(2/d)} } \quad \text{ or } \quad \|w_N\|_1 \ll \dfrac{1}{N(\ln N)^{2/d} } \ ,
\end{equation}
then for all $\zeta > 0$ one has 
\begin{equation*}
  \lim\limits_{N \to \infty} \mathds P \left(\left| \dfrac{n_N^{1,\omega}}{N} - 1 \right| < \zeta \right) = 1\ .
 \end{equation*}
In other words, there is complete BEC in probability into $\varphi_N^{1,\omega}$.
 
 \item[$(\mathrm{ii})$] For any $0 < c < 1$ and $\epsilon > 0$ there exists a constant $\kappa > 0$ such that if
 \begin{equation*}
  \|w_N\|_{\infty} \le \dfrac{\kappa}{N(\ln N)^{1+(2/d)} } \quad \text{ or } \quad \|w_N\|_1 \le \dfrac{\kappa}{N(\ln N)^{2/d} }
 \end{equation*}
for all but finitely many $N \in \mathds N$, one has
\begin{equation*}
  \liminf\limits_{N \to \infty} \mathds P \left( \dfrac{n_N^{1,\omega}}{N} \ge c \right) \ge 1 - \epsilon\ .
 \end{equation*}
 Hence, there is (possibly non-complete) BEC with probability almost one into $\varphi_N^{1,\omega}$.
 \end{enumerate}
\end{theorem}
\begin{proof}
 On the one hand,
 \begin{align}\label{UpperBoundEnergy}
  E_{\text{QM},N}^{1,\omega} & \le \braketop{\varphi_N^{1,\omega}, \ldots, \varphi_N^{1,\omega}}{H_N^{\omega}}{\varphi_N^{1,\omega}, \ldots, \varphi_N^{1,\omega}}\nonumber\\
  & =  N \braketop{\varphi_N^{1,\omega}}{- \Delta}{\varphi_N^{1,\omega}} + \dfrac{N(N-1)}{2}  \braketop{\varphi_N^{1,\omega}, \varphi_N^{1,\omega}}{W_N^{(2)}}{\varphi_N^{1,\omega}, \varphi_N^{1,\omega}}\nonumber \\
  & = N e_N^{1,\omega} + \frac{N(N-1)}{2} \int\limits_{\Lambda_N^{\omega}} \int\limits_{\Lambda_N^{\omega}} w_N(x-y) |\varphi_N^{1,\omega}(x)|^2 |\varphi_N^{1,\omega}(y)|^2 \, \mathrm{d}x \mathrm{d}y\nonumber\\
  & \le N e_N^{1,\omega} + \dfrac{N(N-1)}{2} \min\Big\{ \|w_N\|_{\infty}, C^2 \|w_N\|_1 (e_N^{1,\omega})^{d/2}\Big\} \ ,
 \end{align}
 where we use the estimate $\|\varphi_N^{1,\omega}\|_{\infty} \leq C (e_N^{1,\omega})^{d/4}$ with a suitable constant $C >0$; see, for example, \cite[Lemma 1.1]{SznitmanKAC}.
 On the other hand, since $\Psi_N^{1,\omega}$ is symmetric,
 \begin{align*}
  E^{1,\omega}_{\text{QM},N} &= \Tr{N}_{\jparticlehilbertspaceSymAbrev{N}}\left(H_N^{\omega} \, \densitymatrix{N}\right) \geq \Tr{N}_{\jparticlehilbertspaceAbrev{N}}\left(\sum_{j=1}^{N}(-\Delta_j) \, \densitymatrix{N} \right) =
  N \tr_{\oneparticlehilbertspaceAbrev} \left( (-\Delta) \, \densitymatrix{1} \right) \\
  & = N \tr_{\oneparticlehilbertspaceAbrev} \left( \sum\limits_{j \in \mathds N} e_N^{j,\omega} \projectionvarphiNjomega{j} \, \densitymatrix{1} \right) \ge N \left( \dfrac{n_N^{1,\omega}}{N} e_N^{1,\omega} + \left( 1 - \dfrac{n_N^{1,\omega}}{N} \right) e_N^{2,\omega} \right) \ .
 \end{align*}
 Note that the above calculations with traces involve unbounded operators; for a justification thereof see, for example, \cite[Chapter 3]{LSBook}. Combining the upper and lower bound, we obtain
 \begin{align}\label{ImpEq}
  1 - \dfrac{n_N^{1,\omega}}{N} \le \dfrac{N \min\{ \|w_N\|_{\infty}, C^2 \|w_N\|_1 (e_N^{1,\omega})^{d/2} \}}{e_N^{2,\omega} - e_N^{1,\omega}} \ ,
 \end{align}
 whenever $\omega \in \Omega$ is such that $e_N^{2,\omega} - e_N^{1,\omega} > 0$. 

Now, we prove (i): According to \cite[Chapter 4, Theorem 4.6]{sznitman1998brownian} one knows that, $\mathds P$-almost surely, $e_N^{1,\omega}(\ln N)^{2/d}$ converges to some fixed constant $c_0 > 0$. This implies that $\mathds P \left (e^{1,\omega}_N \leq 2c_0(\ln N)^{-2/d}\right)$ $\longrightarrow 1$ as $N \rightarrow \infty$. Furthermore, \cite[Theorem 6.1]{SznitmanKAC} implies that
 \begin{equation*}
  \lim\limits_{\sigma \to 0} \liminf\limits_{N \to \infty} \mathds P \left (e_N^{2,\omega} - e_N^{1,\omega} \ge \sigma (\ln N)^{-(1+2/d)} \right) = 1 \ .
 \end{equation*}
Now, fix some $\zeta,\epsilon > 0$ and choose $\sigma > 0$ so small such that
\[
\liminf\limits_{N \to \infty} \mathds P \left (e_N^{2,\omega} - e_N^{1,\omega} \ge \sigma (\ln N)^{-(1+2/d)} \right) > 1-\frac{\epsilon}{2}\ .
\]
Then, since $\min\{\|w_N\|_{\infty}, C^2 (2c_0)^{d/2}\|w_N\|_1 (\ln N)^{-1}\} \ll N^{-1} (\ln N)^{-(1+2/d)}$, \eqref{ImpEq} yields, for $N$ large enough,
 \begin{equation*}\begin{split}
  \mathds P \left(\left| \dfrac{n_N^{1,\omega}}{N} - 1 \right| < \zeta \right) &\geq  \mathds P \left (\{e_N^{2,\omega} - e_N^{1,\omega} \ge \sigma (\ln N)^{-(1+2/d)}\} \cap \{e^{1,\omega}_N \leq 2c_0(\ln N)^{-2/d}\}\right)\\
  & \geq  \mathds P \left (e_N^{2,\omega} - e_N^{1,\omega} \ge \sigma (\ln N)^{-(1+2/d)} \right)+\mathds P \left (e^{1,\omega}_N \leq 2c_0(\ln N)^{-2/d}\right)-1 \\
& \geq 1-\frac{\epsilon}{2}- \frac{\epsilon}{2}\\ 
  &= 1- \epsilon\ .
  \end{split}
 \end{equation*}
This readily implies
 \begin{equation*}
  \lim\limits_{N \to \infty} \mathds P \left(\left| \dfrac{n_N^{1,\omega}}{N} - 1 \right| < \zeta \right) = 1
 \end{equation*}
 for all $\zeta > 0$.
 
Consider (ii): Fix arbitrary $0 < c < 1$ and $0 < \epsilon < 1$. We then choose $\sigma > 0$ so small that $\liminf\limits_{N \to \infty} \mathds P \left (e_N^{2,\omega} - e_N^{1,\omega} \ge \sigma (\ln N)^{-(1+2/d)} \right) > 1-\frac{\epsilon}{2}$. We now use \eqref{ImpEq}, and choosing $\kappa > 0$ sufficiently small, we obtain, for $N$ large enough,
 \begin{equation*}\begin{split}
\mathds P \left( \dfrac{n_N^{1,\omega}}{N} \ge c \right)  &\geq  \mathds P \left (\{e_N^{2,\omega} - e_N^{1,\omega} \ge \sigma (\ln N)^{-(1+2/d)}\} \cap \{e^{1,\omega}_N \leq 2c_0(\ln N)^{-2/d}\}\right)\\
&\geq  \mathds P \left (e_N^{2,\omega} - e_N^{1,\omega} \ge \sigma (\ln N)^{-(1+2/d)} \right)+\mathds P \left (e^{1,\omega}_N \leq 2c_0(\ln N)^{-2/d}\right)-1 \\
  & \geq 1-\frac{\epsilon}{2}- \frac{\epsilon}{2}\\ 
  &= 1- \epsilon\ .
  \end{split}
 \end{equation*}
From this the statement follows. 
\end{proof}

Theorem~\ref{TheoremMacOcc} is similar to Theorem 4.2 in \cite{boccato2024interacting}. However, there the minimizer of a Hartree-type functional was used as the one-particle state, rather than the ground state of the Dirichlet Laplacian considered here.~But Theorem 4.2(ii) in \cite{boccato2024interacting} is stronger than the corresponding result presented here, since \cite{boccato2024interacting} establishes complete Bose--Einstein condensation (BEC) with probability almost one, whereas we can only prove non-complete BEC with probability almost one. However, since we want to establish a transition in the type of Bose--Einstein condensation, it is pivotal to include Theorem~\ref{TheoremMacOcc} in this paper. Also, it suggests that the physical properties of the system begin to change when the interaction strength is such that $\min\left\{\|w_N\|_{\infty},\, \|w_N\|_1 (\ln N)^{-1}\right\}$ is of order $ N^{-1} (\ln N)^{-(1+2/d)}$.
\section{Auxiliary results}
\noindent In this section we establish some auxiliary results that are used later to prove our main results in the absence of BEC; our results are inspired by \cite{M07}.
\begin{lemma} \label{Lemma 3}
Let $\varphi_N \in \oneparticlehilbertspaceAbrev$ be a normalized one-particle state and $\Psi_N \in \jparticlehilbertspaceSymAbrev{N}$ a normalized $N$-particle state.  If
\begin{equation}\label{EqProof}
\lim\limits_{N \to \infty}  \Trtwo_{\jparticlehilbertspaceAbrev{2}} \left( \projectionvarphiNtwice \densitymatrixgen{2} \right) = 0\ ,
\end{equation} 
then
\begin{equation*}
\lim\limits_{N \to \infty} \tr_{\oneparticlehilbertspaceAbrev} \left( \projectionvarphiN \densitymatrixgen{1} \right) = 0\ .
\end{equation*}
\end{lemma}
\begin{proof} Suppose \eqref{EqProof} holds. First we show that for any $2 \le J \in \mathds N$, one has
\begin{align}\label{limsup J statement}
  \limsup\limits_{N \to \infty}\ \braketop{\varphi_N}{\densitymatrixgen{1}}{\varphi_N} 
  =  \limsup\limits_{N \to \infty} \sum\limits_{j_2=2}^{\infty} \ldots \sum\limits_{j_J=2}^{\infty} \sum\limits_{j_{J+1}=1}^{\infty} \ldots \sum\limits_{j_N=1}^{\infty} | \braket{\varphi_N, \varphi_N^{j_2}, \ldots, \varphi_N^{j_N}}{\Psi_N}|^2 \ ,
\end{align}
where $(\varphi^{j}_N)_{j \in \mathds{N}}$ is an orthonormal basis of $\oneparticlehilbertspaceAbrev$ such that $\varphi^{1}_N=\varphi_N$. Since $\Psi_N$ is symmetric, we conclude
\begin{align*}
  &\limsup\limits_{N \to \infty} \sum\limits_{j_2=2}^{\infty} \ldots \sum\limits_{j_J=2}^{\infty} \sum\limits_{j_{J+1}=1}^{\infty} \ldots \sum\limits_{j_N=1}^{\infty} | \braket{\varphi_N, \varphi_N^{j_2}, \ldots, \varphi_N^{j_N}}{\Psi_N}|^2 \\
 &= \limsup\limits_{N \to \infty} \Bigg( \sum\limits_{j_2=1}^{\infty} \ldots \sum\limits_{j_N=1}^{\infty} | \langle\varphi_N, \varphi_N^{j_2}, \ldots, \varphi_N^{j_N}|{\Psi_N}\rangle|^2 \\
 & -\sum\limits_{j_{J+1}=1}^{\infty} \ldots \sum\limits_{j_N=1}^{\infty} | \langle\varphi_N, \ldots,\varphi_N,\varphi_N^{j_{J+1}}, \ldots, \varphi_N^{j_N}|{\Psi_N}\rangle|^2 \\
 & - \binom{J-1}{1} \sum\limits_{j_J=2}^{\infty} \sum\limits_{j_{J+1}=1}^{\infty} \ldots \sum\limits_{j_N=1}^{\infty} | \langle\varphi_N, \dots, \varphi_N, \varphi_N^{j_J}, \varphi_N^{j_{J+1}},\ldots, \varphi_N^{j_N}|{\Psi_N}\rangle|^2\\
 & - \binom{J-1}{2} \sum\limits_{j_{J-1}=2}^{\infty} \sum\limits_{j_{J}=2}^{\infty} \sum\limits_{j_{J+1}=1}^\infty\ldots \sum\limits_{j_N=1}^{\infty} | \langle\varphi_N, \dots, \varphi_N, \varphi_N^{j_{J-1}}, \varphi_N^{j_{J}},\varphi_N^{j_{J+1}},\ldots, \varphi_N^{j_N}|{\Psi_N}\rangle|^2 -  \ldots  \Bigg)
 \\
 &= \limsup\limits_{N \to \infty} \Big( \braketop{\varphi_N}{\densitymatrixgen{1}}{\varphi_N} - \braketop{\varphi_N, \ldots, \varphi_N}{\densitymatrixgen{j}}{\varphi_N, \ldots, \varphi_N} 
 \\
 &- \binom{J-1}{1} \braketop{\varphi_N, \dots, \varphi_N}{\densitymatrixgen{J-1}}{\varphi_N, \dots, \varphi_N}  \, - \ldots  \Big) \\
 &= \limsup\limits_{N \to \infty}\ \braketop{\varphi_N}{\densitymatrixgen{1}}{\varphi_N} \ .
 \end{align*}
 In the last step we used that, for $j \geq 2$, 
 \begin{align*}
  \braketop{\varphi_N, \ldots,\varphi_N}{\densitymatrixgen{j}}{\varphi_N, \ldots,\varphi_N}  \leq \braketop{\varphi_N, \varphi_N}{\densitymatrixgen{2}}{\varphi_N, \varphi_N},
 \end{align*}
 and hence all terms above (of which there are finitely many) except $\braketop{\varphi_N}{\densitymatrixgen{1}}{\varphi_N}$ converge, by assumption, to zero in the limit $N \to \infty$. This finishes the proof of \eqref{limsup J statement}. 
 
 Now, assume by contradiction that  $\limsup_{N \to \infty} \tr_{\oneparticlehilbertspaceAbrev} \left( \projectionvarphiN \densitymatrixgen{1}) \right) = c$ for a constant $c > 0$. \eqref{limsup J statement} implies, for any $2 \le J \in \mathds N$, using that $\Psi_N$ is symmetric for all $N \in \mathds N$,
 \begin{align*}
  1 & = \lim\limits_{N \to \infty} \Tr{N}_{\jparticlehilbertspaceAbrev{N}} \left( \densitymatrixgen{N} \right) = \lim\limits_{N \to \infty} \sum\limits_{j_1=1}^{\infty} \ldots \sum\limits_{j_N=1}^{\infty} | \langle\varphi_N^{j_1}, \ldots, \varphi_N^{j_N}|{\Psi_N}\rangle|^2 \\
  & \ge \limsup\limits_{N \to \infty} \Bigg( J \sum\limits_{j_2=2}^{\infty} \ldots \sum\limits_{j_J=2}^{\infty} \sum\limits_{j_{J+1}=1}^{\infty} \ldots \sum\limits_{j_N=1}^{\infty} | \braket{\varphi_N, \varphi_N^{j_2}, \ldots, \varphi_N^{j_N}}{\Psi_N}|^2 \Bigg) \\
  & = J \limsup\limits_{N \to \infty}\ \braketop{\varphi_N}{\densitymatrixgen{1}}{\varphi_N} = Jc\ .
 \end{align*}
Since $J$ was arbitrary, this gives the result.

\end{proof}

Lemma~\ref{Lemma 3} immediately implies the following statement. 

\begin{corollary}\label{Lemma 4}
Let $\varphi_N \in \oneparticlehilbertspaceAbrev$ be a normalized one-particle state and $\Psi_N \in \jparticlehilbertspaceSymAbrev{N}$ a normalized $N$-particle state. If
\begin{equation*}
\limsup\limits_{N \to \infty} \tr_{\oneparticlehilbertspaceAbrev} \left( \projectionvarphiN \densitymatrixgen{1} \right) = c\ ,
\end{equation*}
for a constant $c > 0$, then
\begin{equation*}
\limsup\limits_{N \to \infty}  \Trtwo_{\jparticlehilbertspaceAbrev{2}} \left( \projectionvarphiNtwice \densitymatrixgen{2} \right) = \widetilde c\ ,
\end{equation*}
for a constant $\widetilde c > 0$.
\end{corollary}
In the next statement, we give a lower bound to the number of balls in $\Lambda_N$ that are free of impurities. This will be useful later in the proof of Theorem~\ref{AbsenceofLocalizedBEC}. 
\begin{lemma}[Balls free of Poisson points]\label{Lemma almost sure volume of vacancy set}
Let an arbitrary intensity $\nu > 0$ be given and fix some $r > 0$. Suppose that the constant $c_1 > 0$ is such that $0< 2 C(d) c_1 <\nu^{-1}$ with $C(d)$ being the volume of the unit ball in $\mathds R^d$. Furthermore, let $(m_N^{(1)})_{N \in \mathds N}$ and $(m_N^{(2)})_{N \in \mathds N}$ be two sequences with the following properties:
\begin{enumerate}
\item[]
\item[$ (\mathrm{i})$] $m_N^{(1)} \ge 1$ and $m_N^{(2)} \ge 0$ for all $N \in \mathds N$\ ,
\item[]
\item[$ (\mathrm{ii})$] $m_N^{(1)} \ll N^{1-2C(d) c_1\nu }(2(c_1 \ln N)^{1/d} +2r + m_N^{(2)})^{-d} (\ln N)^{-1}$\ ,
\item[]
\item[$(\mathrm{iii})$] $N^{1-2 C(d) c_1\nu }(2(c_1 \ln N)^{1/d} +2r + m_N^{(2)})^{-d} \gg  (\ln N)^{1+\eta}$ for some $\eta > 0$\ .
\item[]
\end{enumerate}

Then, $\mathds P$-almost surely for all but finitely many $N \in \mathds N$, at least $m_N^{(1)}$ many balls of radius $(c_1 \ln N)^{1/d} + r$ within $\Lambda_N$ are free of Poisson points. Furthermore, these balls have pairwise at least a distance of $m_N^{(2)}$ to each other.
\end{lemma}
\begin{proof}
Let us introduce some abbreviations and define the length
\[
\widetilde g\coloneqq  2(c_1 \ln N)^{1/d} +2r + m_N^{(2)}\,,
\]
the lattice points
\begin{equation*}
   \widetilde G_{r,N} \coloneqq  \ \Big\{- \lfloor \tfrac12 L_N/\widetilde g\rfloor + 1, 
  \ldots, -1+\lfloor \tfrac12 L_N/\widetilde g\rfloor \Big\} \subset \mathds Z\,,
 \end{equation*}
and the centers
\begin{align}\label{Definition of JdrN}
G_{d,r,N} \coloneqq \left\{ \widetilde g\cdot (z_1, \ldots, z_d) \mbox{ with } z_i\in \widetilde G_{r,N} \mbox{ for all } i \right\} \subset \Lambda_N = \left(-\tfrac12 L_N,\tfrac12 L_N\right)^d
 \end{align}
with $L_N = \rho^{-1/d} N^{1/d}$ as above. Then we place $\mathbf{N}\coloneqq \big(2 \lfloor \tfrac12 L_N/\widetilde g\rfloor-2\big)^d \le (L_N/\widetilde{g})^d$ many balls of radius $\widetilde r\coloneqq (c_1 \ln N)^{1/d} + r \leq \widetilde g/2$ with centers at the points in $G_{d,r,N}$. These balls are all within $\Lambda_N$ and have a distance of at least $m_N^{(2)}$ to each other. The volume of such a ball is equal to $C(d) \widetilde r^d$ with the known constant $C(d)$. Below we will need the simple bound 
\begin{align}\label{upper bound in r}
\widetilde r^d  = c_1 \ln N \Big(1+ r (c_1 \ln N)^{-1/d}\Big)^d < 2 c_1 \, \ln N
\end{align}
that holds for $N$ large enough.  

Let $A_N$ be the event that less than $m_N^{(1)}$ of these balls are free of Poisson points. Our assumptions, in particular (ii), imply that $m_N^{(1)}\ll \mathbf{N}$, and one has
\begin{align*}
 \mathds{P}(A_N)&=\sum_{i=0}^{m_N^{(1)}-1} {{\mathbf N}\choose i} \exp\big(-i \nu C(d) \widetilde r^d\big) \cdot \Big[1-\exp\big(-\nu C(d) \widetilde r^d\big)\Big]^{\mathbf{N}-i}
\\
&\le \sum_{i=0}^{m_N^{(1)}-1} {{\mathbf N}\choose i} \cdot 1 \cdot \Big[1-\exp\big(-\nu c_2 \ln N \big)\Big]^{\mathbf{N}-i}\,,
\end{align*}
where we used \eqref{upper bound in r} for large enough $N$ and set $c_2\coloneqq 2 c_1 \cdot C(d)$. For each term in this sum we have the rough bounds 
\begin{align*} {{\mathbf N}\choose i} &\le \mathbf{N}^{m_N^{(1)}} \le \left(\frac{N}{\rho \widetilde g^d}\right)^{m_N^{(1)}}
\\
 \left(1-\exp\big(-\nu c_2 \ln N \big)\right)^{\mathbf{N}-i}&=\big(1-N^{-\nu c_2}\big)^{\mathbf{N}-i} \le \big(1-N^{-\nu c_2}\big)^{\mathbf{N}-m_N^{(1)}} \le \big(1-N^{-\nu c_2}\big)^{\tfrac12 \frac{N}{\rho \widetilde g^d}}\,,
\end{align*}
where we used $\mathbf{N}-m_N^{(1)} \ge \tfrac12 N/(\rho \widetilde g^d)$ for large $N$. As a result, the probability of $A_N$ can be estimated as (using $m_N^{(1)} \le (N/(\rho \widetilde g^d))^{ m_N^{(1)}}$ and $\ln(1-x)\le -x$ for $x\in(0,1)$),
\begin{align*}
    \mathds{P}(A_N)&\le m_N^{(1)} \cdot \left(\frac{N}{\rho \widetilde g^d}\right)^{ m_N^{(1)}}\cdot \big(1-N^{-\nu c_2}\big)^{\frac{N}{2\rho \widetilde g^d}}
    \\&\le \left(\frac{N}{\rho \widetilde g^d}\right)^{ 2 m_N^{(1)}} \cdot \exp\Big(-N^{-\nu c_2} \cdot \frac{N}{2\rho \widetilde g^d}\Big)
    \\
    &= \exp\left(2 m_N^{(1)} \ln\Big(\frac{N}{\rho \widetilde g^d}\Big) - \frac{N^{1-\nu c_2}}{2\rho\widetilde g^d}\right)
    \\
    &\le \exp\big(-C (\ln N)^{1+\eta}\big)  \le N^{-2}
    \end{align*}
for large $N$ and a suitable constant $C > 0$; here, we also used assumptions (ii) and (iii).

Then, the sum $\sum_{N\ge1} \mathds{P}(A_N)$ is finite, and the Lemma of Borel--Cantelli tells us that the event $\limsup A_N \coloneqq \bigcap_{N\ge 1}\bigcup_{k\ge N} A_k$ -- which is the set of points in $\Omega$ that are in infinitely many $A_N$'s -- has probability 0. The complementary event has probability 1 and this is the event that eventually at least $m_N^{(1)}$ balls are free of Poisson points.
This finishes the proof.
\end{proof}

\section{Main results I: Absence of BEC in localized states}\label{SectionAbsence}

\noindent In this section we want to establish that, for strong enough two-particle interactions, there cannot be condensation into a sufficiently localized state. In other words, if there is type-I BEC at all (see the discussion after Theorem~\ref{GenBECWeak}), then the corresponding macroscopically occupied states must be spread over larger domains or fragmented over a number of disconnected components of the vacancy set $\Lambda_N^\omega$ that converges to infinity. In this paper, we cannot exclude these possibilities.

In order to formulate our main result of this section, assume that the $N$-particle system is in the $N$-particle state $\Psi_N^{1,\omega}$ (meaning we are at zero temperature); furthermore, $\varphi_N \in \oneparticlehilbertspaceAbrev$ shall be a one-particle state with (essential) support $\supp(\varphi_N)$.

\begin{theorem}[Absence of localized BEC]\label{AbsenceofLocalizedBEC} Let $\varphi_N \in \oneparticlehilbertspaceAbrev$. Let $m_N^{(1)},m_N^{(2)}$ be sequences as in Lemma~\ref{Lemma almost sure volume of vacancy set} such that, in addition, $m_N^{(1)}$ goes to infinity. Furthermore, assume that the two-particle interaction $w_N$ is such that, with $c_1$ and $G_{d,r,N}$ as in \eqref{Definition of JdrN},
\begin{enumerate}
\item[$(\mathrm{i})$] $\inf \Big\{w_N(x-y) : x,y \in  \supp(\varphi_N)\Big\} \gg \dfrac{1}{N (\ln N)^{2/d}}$\ ,
    \item[$(\mathrm{ii})$]  $\inf \Big\{w_N(x-y) : x,y \in  \supp(\varphi_N)  \Big\} \gg \dfrac{1}{m_N^{(1)}}\min\{\|w_N\|_{\infty}, \|w_N\|_1 (\ln N)^{-1}\}$\ ,
        \item[$(\mathrm{iii})$] $
 \min\{\|w_N\|_{\infty}, \|w_N\|_1 (\ln N)^{-1}\} \gg \\
 \quad \dfrac{1}{m_N^{(1)}} \sum\limits_{x,y \in G_{d,r,N}: x \neq y} \sup\{ w_N(\widetilde x- \widetilde y) : \widetilde x \in B_{(c_1 \ln N)^{1/d}}(x), \widetilde y \in B_{(c_1 \ln N)^{1/d}}(y) \}$ \ .
\end{enumerate}
Then, $\mathds P$-almost surely, $\varphi_N \in \oneparticlehilbertspaceAbrev$ cannot be macroscopically occupied. That is, $\mathds P$-almost surely we have $$\lim\limits_{N \to \infty} \dfrac{n_{\Psi_N^{1,\omega}}^{\varphi_N}}{N} = 0\ .$$

\end{theorem}
\begin{proof}
We prove this theorem by contradiction. Suppose there exists a set $\widetilde \Omega \subseteq \Omega$ with $\mathds P(\widetilde \Omega) > 0$ such that for all $\omega \in \widetilde \Omega$, $\varphi_N$ is macroscopically occupied, that is,
\begin{equation}\label{EqAssump}
\limsup_{N \to \infty} \dfrac{n_{\Psi_N^{1,\omega}}^{\varphi_N}}{N} = \limsup_{N \to \infty} \tr_{\oneparticlehilbertspaceAbrev} \left( \projectionvarphiN \densitymatrix{1} \right) > 0 \ .
\end{equation}
 \textbf{Part I:} We first prove a lower bound to the $N$-particle ground state energy. Let $\omega \in \widetilde \Omega$ be given. Furthermore, we introduce the closed subspaces $\mathcal F \coloneq \mathrm{span}(\varphi_N \otimes \varphi_N)$ and $\mathcal G \coloneq L^2(\supp(\varphi_N) \times \supp(\varphi_N))\cap \jparticlehilbertspaceSymAbrev{2}$ of $\jparticlehilbertspaceSymAbrev{2}$; here, $\mathcal G \subset \jparticlehilbertspaceSymAbrev{2}$ in the sense that the corresponding functions are extended by zero in an obvious way.
 
 We obtain
\begin{align*}
E_{\text{QM},N}^{1,\omega} & = \Tr{N}_{\jparticlehilbertspaceSymAbrev{N}}\left(H_N^{\omega} \, \densitymatrix{N}\right) \\%\ge \Tr{N}_{\jparticlehilbertspaceAbrev{N}}\left( W_N^{(N)} \densitymatrix{N}\right) \\
& \ge \frac{N(N-1)}{2}  \Trtwo_{\jparticlehilbertspaceSymAbrev{2}} \left( W_N^{(2)} \densitymatrix{2} \right) \\ &\ge \frac{N(N-1)}{2}  \Trtwo_{\jparticlehilbertspaceAbrev{2}} \left( P_{\mathcal G} W_N^{(2)} P_{\mathcal G}  \densitymatrix{2} \right) \\
& \ge \frac{N(N-1)}{2}  \Trtwo_{\jparticlehilbertspaceAbrev{2}} \left( P_{\mathcal G} \densitymatrix{2} \right) \cdot \inf \Big\{w_N(x-y) : x,y \in  \supp(\varphi_N)\Big\} \ . 
\end{align*}
Regarding the second-to-last step, note that $P_{\mathcal G}+P_{\mathcal G^{\perp}}=\mathds{1}_{\jparticlehilbertspaceSymAbrev{2}}$, $W_N^{(2)} P_{\mathcal G} = P_{\mathcal G} W_N^{(2)} P_{\mathcal G}$, and $ \Trtwo_{\jparticlehilbertspaceAbrev{2}} ( W_N^{(2)} P_{\mathcal G^{\perp}}  \densitymatrix{2} ) =  \Trtwo_{\jparticlehilbertspaceAbrev{2}} ( P_{\mathcal G^{\perp}} W_N^{(2)} P_{\mathcal G^{\perp}}  \densitymatrix{2} ) \ge 0$ due to \eqref{trace positive operator density matrix is positive}. Regarding the last step, note that we similarly have $\Tr{2}_{\jparticlehilbertspaceAbrev{2}} (P_{\mathcal G}( W_N^{(2)} - \widetilde w_N^{\omega} \mathds{1}_{\jparticlehilbertspaceSymAbrev{2}}) P_{\mathcal G} \densitymatrix{2}) \ge 0$ with \eqref{trace positive operator density matrix is positive}; here $\widetilde w_N^{\omega} \coloneq \inf \{w_N(x-y) : x,y \in  \supp(\varphi_N)\}$ is such that $P_{\mathcal G} W_N^{(2)} P_{\mathcal G} \ge \widetilde w_N^{\omega} P_{\mathcal G}$.

Moreover, since $\mathcal F \subset \mathcal G$ (recall that $P_{\mathcal F}=\projectionvarphiNtwice$) and using \eqref{inequatliy tr F G subspaces}, we obtain 
\begin{align*}
    \Trtwo_{\jparticlehilbertspaceAbrev{2}} \left( P_{\mathcal G} \densitymatrix{2} \right) & \ge \Trtwo_{\jparticlehilbertspaceAbrev{2}} \left( \projectionvarphiNtwice \densitymatrix{2} \right) \ .
\end{align*}
Lastly, by assumption \eqref{EqAssump} and Corollary~\ref{Lemma 4}, one has $\limsup_{N \rightarrow \infty} \Trtwo_{\jparticlehilbertspaceAbrev{2}}( \projectionvarphiNtwice \densitymatrix{2})=c$ for some $c > 0$. In conclusion, there exists a $c>0$ such that
\begin{equation}\label{EqLowerBoundXXX}
    E_{\text{QM},N}^{1,\omega} \ge c \dfrac{N(N-1)}{2} \cdot \inf \Big\{w_N(x-y) :x,y \in  \supp(\varphi_N)\Big\}
\end{equation}
for infinitely many $N \in \mathds N$.

\textbf{Part II:} We now prove an upper bound to the $N$-particle ground state energy $E_{\text{QM},N}^{1,\omega}$ for any typical $\omega \in \Omega$ and all $N \in \mathds{N}$ using Lemma~\ref{Lemma almost sure volume of vacancy set}. More explicitly, in $\Lambda_N$ we place $m_N^{(1)}$ many balls with centers $G_{d,r,N}$ as described in the proof of Lemma~\ref{Lemma almost sure volume of vacancy set} and each with radius $(c_1 \ln N)^{1/d}$ (making these balls completely free of impurities and not only free of Poisson points, that is, free of centers of the impurities). 

In order to construct a suitable $N$-particle trial state, let $\psi_N^{k,\omega}$ denote the normalized ground state of the Dirichlet Laplacian with domain only in the $k$-th ball, $k \in \{1,2,\ldots, m_N^{(1)}\}$. Based on this we consider the symmetric $N$-particle state $\widetilde \Psi_N^{1,\omega} \in \jparticlehilbertspaceSymAbrev{N}$ given by
\begin{equation*}
 \widetilde \Psi_N^{1,\omega} \coloneq \Big| (m_N^{(1)})^{-1/2} \sum\limits_{k=1}^{m_N^{(1)}} \psi_N^{k,\omega}, \ldots, (m_N^{(1)})^{-1/2} \sum\limits_{k=1}^{m_N^{(1)}} \psi_N^{k,\omega} \Big\rangle \ .
\end{equation*}
Using this trial (product) state we obtain the estimate, with some constants $\tilde{c},C > 0$,  
\begin{equation}\label{EqUpperBoundXXX}\begin{split}
 E_{\text{QM},N}^{1,\omega} &\le \braketop{\widetilde \Psi_N^{1,\omega}}{H_N^{\omega}}{\widetilde \Psi_N^{1,\omega}}\\
& \le N \dfrac{\tilde c}{(c_1 \ln N)^{2/d}}  + \dfrac{N^2}{m_N^{(1)}} \min\{\|w_N\|_{\infty}, \|w_N\|_1 C^2  (\ln N)^{-1}\}\\
& + \, \dfrac{N^2}{(m_N^{(1)})^2} \sum\limits_{x,y \in G_{d,r,N}: x \neq y} \sup\{ w_N(\widetilde x- \widetilde y) : \widetilde x \in B_{(c_1 \ln N)^{1/d}}(x), \widetilde y \in B_{(c_1 \ln N)^{1/d}}(y) \} \ .
\end{split}
\end{equation}
Note here that the first term represents the total kinetic energy, the second term represents the energy due to the interaction of particles within the same ball, and the third term represents the energy due to interactions between particles in different balls.

The proof of the statement now follows from a suitable comparison of the derived upper and lower bound; of course, the upper bound has to be larger than the lower bound. However, assuming (i)--(iii), we see that the lower bound would eventually be larger than the upper bound. This can be seen by using \eqref{EqLowerBoundXXX} and \eqref{EqUpperBoundXXX} and comparing each term of the upper bound to the lower bound.
\end{proof}
We remark that conditions (i) and (ii) of Theorem~\ref{AbsenceofLocalizedBEC} can be understood as the requirement that the interaction is sufficiently strong within the support of $\varphi_N$. Condition (iii) of Theorem~\ref{AbsenceofLocalizedBEC}, on the other hand, can be seen as the requirement that the interaction strength is sufficiently weak at larger distances. However, while all three conditions can be fulfilled when the diameter of the support of $\varphi_N$ is relatively small, Theorem~\ref{AbsenceofLocalizedBEC} can also be used to prove absence whenever $\varphi_N$ is supported on multiple components as long as the diameter of each component is relatively small and the number of components remains bounded in the limit $N \to \infty$.
\subsection{Examples}

We now establish two examples that illustrate the meaning of Theorem~\ref{AbsenceofLocalizedBEC} and that will be important later on. Let $\nu > 0$ be sufficiently large so that we are in the non-percolation regime and suppose $\varphi_N \in \oneparticlehilbertspaceAbrev$ is a one-particle state that is supported on only one component of the vacancy set.

In \cite[Theorem~3.1]{kerner_pechmann_2023} it was shown that any component of the vacancy set (all of them are bounded since one is in the non-percolation regime) has diameter of at most $\sim \ln N$ (almost surely and for $N$ large enough). In other words, in this regime one has relatively good bounds on the localization of states supported on only one component of the vacancy set. Furthermore, let the two-particle interaction be given by
\begin{equation*}
    w_N(x):=\begin{cases} v_N\ , \qquad & \text{if } \|x\|_{\mathds R^d} < R_N\ ,\\
    0\ , \qquad & \text{otherwise}\ ,
    \end{cases}
\end{equation*}
for suitable $v_N,R_N$.
\begin{example}\label{ExampleAbsenceI} Choose $R_N=c_3 \ln N$ and $v_N = c_4 (\ln N)^{-d}$ for some constants $c_3,c_4 > 0$. Then $\|w_N\|_1 \sim 1$. Choosing $m^{(2)}_N=2R_N$, we see that condition (iii) in Theorem~\ref{AbsenceofLocalizedBEC} is trivially fulfilled since the right-hand side is zero (more explicitly, the corresponding balls have a distance of at least $m_N^{(2)}$ to each other and the potential $w_N$ has range equal to a half of $m_N^{(2)}$).

Furthermore, choosing $m^{(1)}_N=(\ln N)^\alpha$ with $\alpha>d-1$ as well as $c_3 > 0$ from above large enough, we conclude that conditions (i) and (ii) are also satisfied. Note here that choosing $c_3 > 0$ large enough ensures that the left-hand side in conditions (i) and (ii) equals $v_N$ due to the bounds on the sizes of the components of the vancancy set mentioned above. Therefore, $\mathds P$-almost surely $\varphi_N$ cannot be macroscopically occupied.
\end{example}
As will become clear later, a particularly important example is as follows (see Corollary~\ref{CorGenBEC}).
\begin{example}\label{ExampleAbsenceIII} Now, choose $R_N=c_3 \ln N$ and $v_N=c_4N^{-1}(\ln \ln N)^{-1}$ for some constants $c_3,c_4 > 0$. Then $\|w_N\|_{\infty} \ll \frac{1}{N}$ and $\|w_N\|_1 \sim \frac{(\ln N)^d}{N \ln \ln N}$. Furthermore, assume that $m^{(2)}_N=2R_N$ and $m^{(1)}_N=(\ln N)^\alpha$ with $\alpha > 0$. Then condition (iii) in Theorem~\ref{AbsenceofLocalizedBEC} is trivially fulfilled since the right-hand side is zero for the same reason as in the previous example.
Then, choosing $c_3 > 0$ large enough, we again conclude that conditions (i) and (ii) are satisfied. Therefore,  $\varphi_N \in \oneparticlehilbertspaceAbrev$ is $\mathds P$-almost surely not macroscopically occupied.
\end{example}
\begin{remark}
We remark that also in the percolation regime as well as in the KL-model with soft Poissonian obstacles, the ground state of the one-particle Dirichlet Laplacian is expected to be localized, see \cite{GHKPoisson}. 
\end{remark}
\section{Main results II: Generalized BEC and transition in the type of condensation}\label{SectionGeneralized}

\noindent The first aim of this section is to study generalized Bose--Einstein condensation into the family of (canonical) eigenstates of the Dirichlet Laplacian on $\Lambda^{\omega}_N$. The key point here is that, since the concept of generalized BEC refers to a somewhat weaker notion of BEC, it can be proved for stronger two-particle interactions. In other words, while for stronger interactions the ground-state of the Dirichlet Laplacian might not be macroscopically occupied -- for example, in the non-percolation regime as demonstrated in Theorem~\ref{AbsenceofLocalizedBEC} and Examples~\ref{ExampleAbsenceI} and \ref{ExampleAbsenceIII} -- generalized BEC might still occur. Here, let us recall that the notion of generalized Bose--Einstein condensation dates back to a paper of Girardeau~\cite{girardeau1960relationship}, see also \cite{SchultzBEC}. Informally speaking, generalized condensation refers to the macroscopic occupation of a family of states (typically consisting of states of an arbitrarily small energy window) as a \textit{whole} (see Theorem~\ref{GenBECWeak} for a precise formulation in our setting). It is therefore possible to have generalized BEC without any single state of the family being macroscopically occupied. On the other hand, however, if one considers the family of eigenstates of the one-particle Laplacian and if -- for example -- the one-particle ground state is macroscopically occupied (as in Theorem~\ref{TheoremMacOcc}), then this immediately implies generalized BEC. For more on generalized condensation we refer to \cite{van1982generalized,BLP} and references therein. To the best of our knowledge, there are only few results on generalized condensation in \textit{interacting} systems of bosons; in the non-random setting, let us refer to \cite{SutoGen} and in the random setting to \cite{kerner2019bose} where generalized BEC was established for interacting bosons in the Luttinger--Sy model. 

A main result in this context is the following statement, see also \cite{BKPSMFO}.
\begin{theorem}[Generalized BEC]\label{GenBECWeak} Assume the system is in the state $\N$, where $\Psi_N^{1,\omega}$ is a ground state of $H_{N}^{\omega}$ defined in  \eqref{N particle Hamiltonian}. Let $n_N^{j,\omega}$ be the occupation number of $\varphi_N^{j,\omega}$. Furthermore, let 
\begin{equation*}
\|w_N\|_{\infty} \ll \dfrac{1}{N} \quad \text{ or } \quad \|w_N\|_1 \ll \dfrac{\ln N}{N} \ .
\end{equation*}
Then, $\mathds P$-almost surely,
\begin{equation*}
\lim\limits_{\epsilon \searrow 0} \liminf_{N \to \infty} \dfrac{1}{N} \sum\limits_{j \in \mathds N : e_N^{j,\omega} \le \epsilon} n_N^{j,\omega} = 1 \ ,
\end{equation*}
that is, one has complete g-BEC (into the canonical eigenstates of the one-particle Dirichlet Laplacian on $\Lambda_N^{\omega}$).
\end{theorem}
\begin{proof} 
   First, recall the upper bound \eqref{UpperBoundEnergy}. Next, regarding a lower bound, we proceed in a similar fashion as in the proof of Theorem~\ref{TheoremMacOcc} to obtain
    \begin{align*}
   % \begin{split}
        E^{1,\omega}_{\text{QM},N} &= \Tr{N}_{\jparticlehilbertspaceSymAbrev{N}}\left(H_N^{\omega} \, \densitymatrix{N}\right) \geq \Tr{N}_{\jparticlehilbertspaceSymAbrev{N}}\left(\sum_{j=1}^{N}(-\Delta_j) \, \densitymatrix{N} \right)\\
        &= N \, \tr_{\oneparticlehilbertspaceAbrev}\left((-\Delta) \, \densitymatrix{1} \right) = N \, \tr_{\oneparticlehilbertspaceAbrev}\left( \sum\limits_{j=1}^{\infty} e_N^{j,\omega} \projectionvarphiNjomega{j} \, \densitymatrix{1} \right) \\
        &% \ge N \left( e_N^{1,\omega} \tr_{\oneparticlehilbertspaceAbrev} \left( \projectionvarphiNjomega{1} \densitymatrix{1} \right) + \tr_{\oneparticlehilbertspaceAbrev}\left(  e_N^{2,\omega} \sum\limits_{j=2}^{\infty} \projectionvarphiNjomega{j} \, \densitymatrix{1} \right) \right)
        \geq N \, \left( e^{1, \omega}_{N} \sum_{j=1}^{J_N^{\omega}-1} \frac{n^{j,\omega}_N}{N}+ e^{J_N^{\omega},\omega}_N \, \frac{N-\sum_{j=1}^{J_N^{\omega}-1}n^{j,\omega}_N}{N}\right) \ ,
    %\end{split}
    \end{align*}
    where
    \begin{equation*}
    J_N^{\omega} := \min \Bigg\{j \in \mathds N : e_N^{j,\omega} \ge 2 \max\Big\{ N \gamma_N \min \big\{ \|w_N\|_{\infty}, C^2 \|w_N\|_1 (e_N^{1,\omega})^{d/2} \big\}, e_N^{1,\omega} \Big\} \Bigg\}\,,
    \end{equation*}
    and $(\gamma_N)_{N \in \mathds{N}}$ is a positive sequence such that $\gamma_N \gg 1$ and
    \begin{equation}\label{EqXXX}
    N  \gamma_N \min \big\{ \|w_N\|_{\infty},C^2\|w_N\|_1 (\ln N)^{-1} \big\} \ll 1 \ .
    \end{equation}
Here, note that $e_N^{1,\omega} \sim (\ln N)^{-2/d}$ $\mathds P$-almost surely, see for example \cite[Chapter 4, Theorem 4.6]{sznitman1998brownian}.
    
    Now, combining our upper and lower bound, we get
 \begin{align*}
     N \min \Big\{ \|w_N\|_{\infty},  C^2 \|w_N\|_1 (e_N^{1,\omega})^{d/2} \Big\} & \geq %e^{1, \omega}_{N} \left( \sum_{j=1}^{J_N^{\omega} -1}\frac{n^{j,\omega}_N}{N}-1\right)+ e^{J_N^{\omega},\omega}_N \left(1-\sum_{j=1}^{J_N^{\omega} -1}\frac{n^{j,\omega}_N}{N}\right) = \\
     \left(1-\sum_{j=1}^{J_N^{\omega} -1} \frac{n^{j,\omega}_N}{N}\right)\left( e^{J_N^{\omega},\omega}_N -e^{1,\omega}_N \right) \ ,
     \end{align*}
    and therefore, with~\eqref{EqXXX}, 
     \begin{equation}\label{EQPROOF1}\begin{split}
   \left(1 - \frac{1}{N} \sum_{j=1}^{J_N^{\omega} -1} n^{j,\omega}_N \right) &\leq \frac{ N \min \Big\{ \|w_N\|_{\infty},  C^2 \|w_N\|_1 (e_N^{1,\omega})^{d/2} \Big\}}{e^{J_N^{\omega},\omega}_N - e^{1,\omega}_N} \le \frac{1}{\gamma_N} \ .
    \end{split}
    \end{equation}
    Now, the right-hand side of~\eqref{EQPROOF1} converges to zero, and $e_N^{J_N^{\omega}-1, \omega}$ converges $\mathds P$-almost surely to zero, due to the definition of $J_N^{\omega}$. Hence, $\mathds P$-almost surely,
    \begin{equation*}
    \lim_{\varepsilon \searrow 0}\liminf_{N \rightarrow \infty}\frac{1}{N}\sum_{j \in \mathds N: e^{j,\omega}_N \leq \varepsilon }n^{j,\omega}_N \geq \liminf_{N \rightarrow \infty}\frac{1}{N}\sum_{j=1}^{J_N^{\omega} - 1}n^{j,\omega}_N=1\ .
    \end{equation*}
\end{proof}
An interesting notion related to generalized Bose--Einstein condensation is that of a type-III BEC. We say that type-III BEC (into a family of eigenstates of the one-particle Laplacian) is present if one has generalized BEC as in Theorem~\ref{GenBECWeak} and if each individual eigenstate is itself \textit{not} macroscopically occupied. In addition, we say that BEC is of type I if only finitely many eigenstates are macroscopically occupied (of course, type-I BEC implies generalized BEC). 

This definition now yields an important application of Theorem~\ref{GenBECWeak}, obtained in connection with Example~\ref{ExampleAbsenceIII}.
\begin{corollary}[Type-III BEC]\label{CorGenBEC} Choose the intensity $\nu > 0$ of the Poisson process so large that one is in the non-percolation regime. Assume that the system is in the state $\N$, where $\Psi_N^{1,\omega}$ is a ground state of $H_{N}^{\omega}$ defined in  \eqref{N particle Hamiltonian}. Then, choosing the two-particle interaction $w_N:\mathds{R}^d \rightarrow \mathds{R}_+$ as in Example~\ref{ExampleAbsenceIII}, one has type-III Bose--Einstein condensation into the family of (canonical) eigenstates of the one-particle Dirichlet Laplacian. 
\end{corollary}
In words, Example~\ref{ExampleAbsenceIII} provides us with specific two-particle interactions for which the ground state of the (one-particle) Dirichlet Laplacian -- as well as any other canonical one-particle eigenstate -- is not macroscopically occupied in the non-percolation regime, but there is generalized BEC into the (canonical) eigenstates of the Dirichlet Laplacian. Consequently, this proves a transition from a type-I BEC (which exists in the non-interacting or weakly interacting regime as described by Theorem~\ref{TheoremMacOcc}(i)) to a type-III BEC, similar to what has been observed in \cite{kerner2019bose} for the random one-dimensional Luttinger--Sy model. 

\section{Final remarks}\label{SectionDiscussion}

\noindent We now compare the results obtained in this paper with previous results, see also \cite{BKPSMFO}. In \cite{boccato2024interacting} and similar to Theorem~\ref{TheoremMacOcc}, the authors proved a macroscopic occupation (with probability almost one) of a one-particle state, the minimizer of a Hartree-type functional, in the Kac--Luttinger model for certain interaction potentials including $w_N:\mathds{R}^d \rightarrow \mathds{R}_+$ of the form
\begin{equation}\label{OldModelPotential}
    w_N(x)= \frac{\kappa}{N(\ln N)^{2/d}}W(x) \ ,
\end{equation}
where $W \in (L^1 \cap L^{\infty})(\mathds{R}^d,\mathds{R}_+)$ is a suitable potential (satisfying some extra assumptions that we can neglect here) and $\kappa > 0$ is sufficiently small; note that $\| w_N\|_1 \sim \frac{1}{N(\ln N)^{2/d} }$. Theorem~\ref{GenBECWeak} now shows that generalized BEC is present, for example, for an interaction potential of the form
\begin{equation}\label{ModelPotential}
    w_N(x)=\frac{(\ln N)^{1-\epsilon}}{N}W(x) 
\end{equation}
with $W \in (L^1 \cap L^{\infty})(\mathds{R}^d,\mathds{R}_+)$ and any $\epsilon > 0$; here, one satisfies the $L^{1}$-condition in Theorem~\ref{GenBECWeak}. In other words, g-BEC can be proved for interactions with a higher power of $\ln N$.

Furthermore, in \cite{kerner_pechmann_2023} the authors investigate absence of macroscopic occupation of sufficiently localized states for interaction potentials of sufficient strength, comparable to what has been done in Theorem~\ref{AbsenceofLocalizedBEC}. While the setting discussed in \cite{kerner_pechmann_2023} is fairly similar to the one discussed here, it is important to note that \cite{kerner_pechmann_2023} assumes a positive temperature whereas we assume zero temperature, and their results are restricted to the non-percolation regime. For example, consider
\begin{equation*}
w_N(x) = c_NW(\|x\|_{\mathds{R}^d}) 
\end{equation*}
with $c_N > 0$ and a bounded continuous and non-negative $W \in L^1(\mathds{R}_+, x^{d-1}\mathrm{d}x)$; note that $\|w_N\|_1 \sim c_N$. For such a potential, it has been proved in \cite[Theorem~4.2]{kerner_pechmann_2023} that no state localized on only one component of the vacancy set (as an example, one could think of the ground state of the one-particle Dirichlet Laplacian) can be macroscopically occupied whenever
\begin{equation*}
    \dfrac{(\ln N)^{3}}{N} \ll c_N \ll \dfrac{1}{(\ln N)^{2}} \ .
\end{equation*}
Consequently, again comparing the potentials \eqref{OldModelPotential} and \eqref{ModelPotential}, we see that the extra powers of $\ln N$ are not negligible since a highly localized condensate can be excluded, at least when the temperature is non-zero, already for $c_N=(\ln N)^{3+\epsilon} / N$ with any $\epsilon > 0$. 

It is also worth comparing Theorem~\ref{GenBECWeak} with results obtained in \cite{kerner2019bose}. In this paper, the authors discuss interacting Bose gases in the one-dimensional Luttinger--Sy model, which can be considered the lower-dimensional analogue of the Kac--Luttinger model. To be more explicit, they consider contact interactions between the particles of the form
\begin{equation*}
    w_N(x)=g_N\delta(x) ,
\end{equation*}
where $g_N > 0$ describes the interaction strength; note that $g_N$ can be understood as the $L^1$-norm of $w_N$. In \cite[Theorem~3.1]{kerner2019bose} the authors then prove generalized condensation (into a family of Gross--Pitaevskii minimizers) for interactions strengths that satisfy
\begin{equation*}
g_N \ll \dfrac{1}{N^{\eta}(\ln N)^{2}}
\end{equation*}
 for an $0 < \eta \leq 1/3$. Most importantly, however, in \cite[Theorem~3.3]{kerner2019bose} a transition in the type of condensation is established: whereas for interaction strengths
 \begin{equation*}
 g_N \ll \dfrac{1}{N(\ln N)^{2}}
 \end{equation*}
 one almost surely has BEC of type I or type II (implying a macroscopic occupation of at least one one-particle state), BEC is almost surely of type III whenever
 \begin{equation*}
 \dfrac{1}{N(\ln N)} \ll g_N \ll \dfrac{1}{N^{\eta}(\ln N)^{2} }
 \end{equation*}
 with an $0 < \eta \leq 1/3$. Translating this to the higher-dimensional Kac--Luttinger model, and, in particular, recalling \eqref{Assumption interaction strength type I BEC KLM}, it was asked in \cite{BKPSMFO} whether there exists a $\gamma \leq 2/d$ such that type-III condensation occurs for interaction potentials $w_N$ considered in Theorem~\ref{GenBECWeak} that additionally satisfy
\begin{equation*}
\|w_N\|_1 \gg \dfrac{1}{ N(\ln N)^{\gamma}} \ .
\end{equation*}
This has now been answered in the positive by Corollary~\ref{CorGenBEC}.
\subsection*{Acknowledgement}{This research was
supported through the program \textit{Oberwolfach Research Fellows} by the Mathematisches Forschungsinstitut Oberwolfach in 2024 and by a \textit{Research in Residence} at the Centre International de Rencontres Mathématiques (CIRM) in 2026. We are grateful to both institutions and warmly thank for their hospitality.

C.B. acknowledges funding from the Italian Ministry of University and Research and Next Generation EU through the PRIN 2022 project
PRIN202223CBOCC\_01, project code 2022AKRC5P. C.B. also acknowledges
support of Grant PID2024-156184NB-I00 funded by MICIU/AEI/10.13039/
501100011033 and cofunded by the European Union. C.B. acknowledges also GNFM (Gruppo Nazionale per la Fisica Matematica) - INDAM.}

	\providecommand{\bysame}{\leavevmode\hbox to3em{\hrulefill}\thinspace}
\providecommand{\MR}{\relax\ifhmode\unskip\space\fi MR }
\providecommand{\MRhref}[2]{%
  \href{http://www.ams.org/mathscinet-getitem?mr=#1}{#2}
}
\providecommand{\href}[2]{#2}


\begin{thebibliography}{BKPS25}

\bibitem[BBCS19]{BBCS2019}
C.~Boccato, C.~Brennecke, S.~Cenatiempo, and B.~Schlein, \emph{{Bogoliubov theory in the {G}ross--{P}itaevskii limit}}, {Acta Math.} \textbf{222} ({2019}), 219--335.

\bibitem[BBCS20]{BBCSoptimalrate}
C.~Boccato, C.~Brennecke, S.~Cenatiempo, and B.~Schlein, \emph{{Optimal Rate for Bose–Einstein Condensation in the Gross–Pitaevskii Regime}}, Commun. Math. Phys. \textbf{376} (2020), no.~3, 1311--1395.

\bibitem[BKP24]{boccato2024interacting}
C.~Boccato, J.~Kerner, and M.~Pechmann, \emph{Interacting many-particle systems in the random {K}ac--{L}uttinger model and proof of {B}ose--{E}instein condensation}, J. Math. Pures Appl. \textbf{189} (2024), 103594.

\bibitem[BKPS25]{BKPSMFO}
C.~Boccato, J.~Kerner, M.~Pechmann, and W.~Spitzer, \emph{Generalized {B}ose--{E}instein condensation in the {K}ac-{L}uttinger model}, Oberwolfach Preprints \textbf{06} (2025).

\bibitem[BL82]{van1982generalized}
M.~van~den Berg and J.~T. Lewis, \emph{On generalized condensation in the free boson gas}, Physica A: Statistical Mechanics and its Applications \textbf{110} (1982), no.~3, 550--564.

\bibitem[BLP86]{BLP}
M.~van~den Berg, J.~T. Lewis, and J.~V. Pul\'{e}, \emph{A general theory of {B}ose--{E}instein condensation}, Helv. Phys. Acta \textbf{59} (1986), no.~8, 1271--1288.

\bibitem[BV25]{BaiVogel}
T.~Bai and Q.~Vogel, \emph{{Gibbs measures for the repulsive Bose gas}}, Annales de l'Institut Henri Poincaré, Probabilités et Statistiques \textbf{61} (2025), no.~3, 2149 -- 2183.

\bibitem[Ein24]{EinsteinBECI}
A.~Einstein, \emph{Quantentheorie des einatomigen idealen {G}ases}, Sitzber. {K}gl. {P}reuss. {A}kadm. {W}iss. (1924), 261--267.

\bibitem[Ein25]{EinsteinBECII}
\bysame, \emph{Quantentheorie des einatomigen idealen {G}ases. {Z}weite {A}bhandlung}, Sitzber. {K}gl. {P}reuss. {A}kadm. {W}iss. (1925), 3--14.

\bibitem[GHK07]{GHKPoisson}
F.~Germinet, P.~D. Hislop, and A.~Klein, \emph{Localization for {S}chrödinger operators with {P}oisson random potential}, J. Eur. Math. Soc. \textbf{9} (2007), 577--607.

\bibitem[Gir60]{girardeau1960relationship}
M.~Girardeau, \emph{Relationship between systems of impenetrable bosons and fermions in one dimension}, J. Math. Phys. \textbf{1} (1960), no.~6, 516--523.

\bibitem[JPZ10]{jaeck2010nature}
T.~Jaeck, J.~V. Pul{\'e}, and V.~Zagrebnov, \emph{On the nature of {B}ose--{E}instein condensation enhanced by localization}, J. Math. Phys. \textbf{51} (2010), no.~10, 103302.

\bibitem[KL73]{kac1973bose}
M.~Kac and J.~M. Luttinger, \emph{{{B}ose--{E}instein condensation in the presence of impurities}}, J.~Math.~Phys. \textbf{14} (1973), no.~11, 1626--1628.

\bibitem[KL74]{kac1974bose}
\bysame, \emph{{{B}ose--{E}instein condensation in the presence of impurities. II}}, J.~Math.~Phys. \textbf{15} (1974), no.~2, 183--186.

\bibitem[KP21]{kerner2021effect}
J.~Kerner and M.~Pechmann, \emph{On the effect of repulsive pair interactions on {B}ose--{E}instein condensation in the {L}uttinger--{S}y model}, Proceedings of the American Mathematical Society \textbf{149} (2021), no.~8, 3499--3513.

\bibitem[KP23]{kerner_pechmann_2023}
\bysame, \emph{Bose--{E}instein condensation for particles with repulsive short-range pair interactions in a {P}oisson random external potential in $\mathbb{R}^{d}$}, J. Appl. Probab. \textbf{60} (2023), no.~2, 382--393.

\bibitem[KPS19]{kerner2019bose}
J.~Kerner, M.~Pechmann, and W.~Spitzer, \emph{{B}ose--{E}instein condensation in the {L}uttinger--{S}y model with contact interaction}, Ann. Henri Poincaré, vol.~20, Springer, 2019, pp.~2101--2134.

\bibitem[KPS20]{KERNER2020287}
\bysame, \emph{On a condition for type-{I} {B}ose--{E}instein condensation in random potentials in d dimensions}, J.~Math.~Pures Appl. \textbf{143} (2020), 287--310.

\bibitem[LS73]{LuttSyEnergy}
J.~M. Luttinger and H.~K. Sy, \emph{{Low-lying energy spectrum of a one-dimensional disordered system}}, Phys. Rev. A \textbf{7} (1973), 701--712.

\bibitem[LS02]{LiebSeiringer2002}
E.~H. Lieb and R.~Seiringer, \emph{{Proof of {B}ose--{E}instein condensation for dilute trapped gases}}, {Phys. Rev. Lett.} \textbf{88} ({2002}), 170409.

\bibitem[LS10]{LSBook}
\bysame, \emph{The stability of matter in quantum mechanics}, Cambridge University Press, 2010.

\bibitem[LVZ03]{ZagLauVer}
J.~Lauwers, A.~Verbeure, and V.~A. Zagrebnov, \emph{Proof of {B}ose–{E}instein condensation for interacting gases with a one-particle spectral gap}, Journal of Physics A: Mathematical and General \textbf{36} (2003), no.~11, L169.

\bibitem[LZ07]{LenZag}
O.~Lenoble and V.~A. Zagrebnov, \emph{Bose--{E}instein condensation in the {L}uttinger--{S}y model}, Markov Process. Related Fields \textbf{13} (2007), no.~2, 441--468.

\bibitem[Mic07]{M07}
A.~Michelangeli, \emph{{Reduced density matrices and {B}ose--{E}instein condensation}}, SISSA \textbf{39} (2007).

\bibitem[MR96]{meester1996continuum}
R.~Meester and R.~Roy, \emph{Continuum percolation}, vol. 119, Cambridge University Press, 1996.

\bibitem[Sch63]{SchultzBEC}
T.~D. Schultz, \emph{Note on the one-dimensional gas of impenetrable point-particle bosons}, J. Mathematical Phys. \textbf{4} (1963), 666--671.

\bibitem[Sch22]{SchleinRev}
B.~Schlein, \emph{{Bose gases in the Gross–Pitaevskii limit: A survey of some rigorous results, in "The Physics and Mathematics of Elliott Lieb"}}, vol.~II, EMS Press, 2022.

\bibitem[Sei14]{SeiRev}
R.~Seiringer, \emph{{Bose gases, Bose–Einstein condensation, and the Bogoliubov approximation}}, J. Math. Phys. \textbf{55} (2014), 075209.

\bibitem[SW16]{SeiringerWarzel}
R.~Seiringer and S.~Warzel, \emph{Decay of correlations and absence of superfluidity in the disordered {T}onks–{G}irardeau gas}, New Journal of Physics \textbf{18} (2016), no.~3, 035002.

\bibitem[Szn98]{sznitman1998brownian}
A.-S. Sznitman, \emph{Brownian motion, obstacles and random media}, Springer Science \& Business Media, 1998.

\bibitem[Szn23]{SznitmanKAC}
\bysame, \emph{On the spectral gap in the {K}ac--{L}uttinger model and {B}ose--{E}instein condensation}, Stoch. Process. Their Appl. (2023), 104197.

\bibitem[Sü04]{SutoGen}
A.~Sütő, \emph{Normal and generalized {B}ose condensation in traps: one dimensional examples}, J. Statist. Phys. \textbf{117} (2004), no.~1-2, 301--341.

\bibitem[Wei12]{weidmann2012linear}
J.~Weidmann, \emph{Linear operators in {H}ilbert spaces}, Springer Science \& Business Media, 2012.

\bibitem[Zag25]{ZagrebnovReviewBEC}
V.~Zagrebnov, \emph{{A century of the Bose-Einstein condensation concept and half a century of the JINR experiments for observation of condensate in the superfluid 4He (He II)}}, 2025, arXiv:2510.03378.

\end{thebibliography}
\end{document}